\documentclass[11pt]{article}
\usepackage{amsmath,amsthm,amsfonts,amssymb}
\usepackage{fullpage}
\usepackage{liyang}

\usepackage{framed}

\usepackage[usenames,dvipsnames]{xcolor}
\makeatletter
\newtheorem*{rep@theorem}{\rep@title}
\newcommand{\newreptheorem}[2]{
\newenvironment{rep#1}[1]{
 \def\rep@title{#2 \ref{##1}}
 \begin{rep@theorem}\itshape}
 {\end{rep@theorem}}}
\makeatother
\def\colorful{1}

\ifnum\colorful=1

\newcommand{\orange}[1]{{\color{orange}{#1}}}
\newcommand{\blue}[1]{{{#1}}}

\fi
\ifnum\colorful=0

\newcommand{\orange}[1]{{{#1}}}
\newcommand{\blue}[1]{{{#1}}}

\fi

\usepackage{boxedminipage}

\newcommand{\ignore}[1]{}

\newcommand{\rnote}[1]{\footnote{{\bf \color{orange}Rocco:} {#1}}}

\newcommand{\down}{\downarrow}
\newcommand{\up}{\uparrow}
\newcommand{\score}{\mathrm{score}}

\renewcommand{\epsilon}{\eps}

\newreptheorem{theorem}{Theorem}
\newreptheorem{lemma}{Lemma}
\newreptheorem{proposition}{Proposition}

\begin{document}
\title{New algorithms and lower bounds for monotonicity testing}

\author{
Xi Chen\thanks{Supported by NSF grant CCF-1149257 and a Sloan research fellowship. } \\ Columbia University \and 
Rocco A.\ Servedio\thanks{Supported by NSF grants CCF-1115703 and CCF-1319788.} \\ Columbia University \\ \and 
Li-Yang Tan\thanks{This work was done while the author was at Columbia Univerisity, supported by NSF grants CCF-1115703 and CCF-1319788.}\\ Simons Institute, UC Berkeley \\ \\ {\tt \{xichen,rocco,liyang\}@cs.columbia.edu}}

\maketitle

\begin{abstract}
We consider the problem of testing whether an unknown Boolean function
$f\isafunc$ is monotone versus $\eps$-far from every monotone
function.  The two main results of this paper are a new lower bound
and a new algorithm for this well-studied problem.

{\bf Lower bound:}
We prove an {$\tilde{\Omega}(n^{1/5})$} lower bound on the query complexity of any non-adaptive two-sided error algorithm for testing whether an unknown
Boolean function $f$ is monotone versus constant-far from monotone.  This gives an exponential improvement on the previous lower bound of $\Omega(\log n)$ due to Fischer {\it et al.}~\cite{FLN+02}. We show that the same lower bound holds for monotonicity testing of Boolean-valued functions over hypergrid
domains $\{1,\dots,m\}^n$ for all $m \geq 2.$

{\bf Upper bound:}
We give an $\tilde{O}(n^{5/6})\poly(1/\eps)$-query
algorithm that tests
whether an unknown Boolean function $f$ is monotone versus $\eps$-far
from monotone.  Our algorithm, which is non-adaptive and makes one-sided error,
is a modified version of the algorithm of Chakrabarty and Seshadhri
\cite{CS13a}, which makes $\tilde{O}(n^{7/8})\poly(1/\eps)$ queries.

\end{abstract}
\thispagestyle{empty}

\newpage
\setcounter{page}{1}

\section{Introduction}

Monotonicity is a basic and natural property of functions. In the field of property testing, the problem of efficiently testing whether an unknown function is monotone has been the focus of a long and fruitful line of research, with many works (see e.g.~\cite{GGL+98, DGL+99, GGL+00, EKK+00, FLN+02, Fis04, BKR04, ACCL07,HK08, RS09c, BBM12, BCGM12, RRSW12, CS13a, CS13b, CS13c, BRY13}) studying this problem for functions with various domains and ranges.

In this work we will be concerned with the classical problem of testing monotonicity of \emph{Boolean functions} $f\isafunc$, which was first posed and considered explicitly by Goldreich {\it et al.}~\cite{GGL+98}.  Recall that a Boolean function $f$ is monotone if $f(x) \le f(y)$ for all $x \prec y$, where $\prec$ denotes the bitwise partial order on the hypercube.  Let $\dist(f,g) := \Pr_{\bx \in \{-1,1\}^n}[f(\bx) \neq g(\bx)]$; we say that $f$ is \emph{$\eps$-close to monotone} if $\dist(f,g) \leq \eps$
for some monotone Boolean function $g$,  and that $f$ is \emph{$\eps$-far from monotone} otherwise.  We will be interested in query-efficient randomized testing algorithms for the following task:
\begin{quote}
{\it Given as input a distance parameter $\eps > 0$ and oracle access to an unknown Boolean function $f\isafunc$, output {\sf Yes} with probability at least $2/3$ if $f$ is monotone, and {\sf No} with probability at least $2/3$ if $f$ is $\eps$-far from monotone.}
\end{quote}

The work of Goldreich {\it et al.}~\cite{GGL+98} proposed a simple ``edge tester'' which queries uniform random edges of $\{-1,1\}^n$ hoping to find an edge whose endpoints violate monotonicity.  \cite{GGL+98} proved an $O(n^2\log(1/\eps)/\eps)$ upper bound on the query complexity of the edge tester, which was subsequently improved to $O(n/\eps)$ in the journal version~\cite{GGL+00}. Fischer {\it et al.}~\cite{FLN+02} established the first lower bounds shortly after, showing that there exists a constant distance parameter $\eps_0 > 0$ such that $\Omega(\log n)$ queries are necessary for any \emph{non-adaptive} tester (one whose queries do not depend on the oracle's responses to prior queries).  This directly implies an $\Omega(\log \log n)$ lower bound for adaptive testers, since any $q$-query adaptive tester can be simulated by a non-adaptive one that simply carries out all $2^q$ possible executions.  These upper and lower bounds were the best known for more than a decade, until the recent work of Chakrabarty and Seshadhri~\cite{CS13a} improved on the linear upper bound of Goldreich {\it et al.}~with an $\tilde{O}(n^{7/8}\eps^{-3/2})$-query tester.

Our main contributions in this work are (i) a new lower bound that improves
on the \cite{FLN+02} lower bound by an exponential factor, and (ii)
a new algorithm that improves on the \cite{CS13a} upper bound (in terms of the dependence on $n$) by a polynomial factor.  We now describe these contributions in more detail.

\medskip
{\bf Our lower bound.} We give an exponential improvement on the
above-mentioned lower bounds of Fischer {\it et al.}:
\begin{theorem}
\label{main-theorem-lb}
There exists a universal constant $\eps_0 > 0$ such that any non-adaptive algorithm for testing whether an unknown Boolean function is monotone versus $\eps_0$-far from monotone must make $\Omega(n^{1/5}(\log n)^{-2/5})$ queries.  Consequently, any adaptive algorithm must make $\Omega(\log n)$ queries.
\end{theorem}

While the aforementioned results of Fischer {\it et al.}~represent the previous best lower bounds on the general testing problem as defined above, additional lower
bounds are known for several restricted versions of the problem.\ignore{there has been significant work on several of its relaxations and variants.}  In the same paper Fischer {\it et al.}~gave an $\Omega(\sqrt{n})$ lower bound on the query complexity of any non-adaptive \emph{one-sided} tester, i.e.~one that always outputs {\sf Yes} when $f$ is monotone (again, this directly implies an $\Omega(\log n)$ lower bound for adaptive one-sided testers).  Restricting further, a \emph{pair tester} is a non-adaptive one-sided tester that independently draws pairs of comparable points $x \prec y$ from some distribution and
rejects if and only if  some pair that is drawn violates monotonicity. Bri\"et {\it et al.}~\cite{BCGM12} proved an $\Omega(n/(\eps\log n))$ lower bound on the query complexity of pair testers whose query complexity can be written as $q(n)/\eps$ for some function $q.$

In addition to Theorem \ref{main-theorem-lb},
we show that essentially the same lower bound holds for monotonicity testing of
Boolean-valued functions
over hypergrid domains $\{1,\dots,m\}^n$ for $m \geq 2.$
(Below and throughout this paper we write $[m]$ to denote $\{1,2,\dots,m\}$.)
Our most general  lower bound is the following:

\begin{theorem}
\label{main-hypergrid}
There exists a universal constant $\eps_0 > 0$ such that for all $m\ge 2$, any non-adaptive algorithm for testing whether an unknown function $f:[m]^n\to \{-1,1\}$ is monotone versus $\eps_0$-far from monotone must make
$\tilde{\Omega}(n^{1/5})$ queries.
\end{theorem}
To the best of our knowledge
Theorem \ref{main-hypergrid} is the first lower bound for testing monotonicity of \emph{Boolean-valued} functions over hypergrid domains. Recent papers of Chakrabarty and Sesha\-dhri \cite{CS13b,CS13c} and Blais {\it et al.}~\cite{BRY13} essentially close the problem of testing monotonicity~of functions $f:[m]^n \to\N$, showing that $\Theta(n\log m)$ queries are both necessary and sufficient; however, their lower bounds crucially depend on the functions considered having range $\N$ rather than $\{-1,1\}$.  \ignore{Finally, we mention that the work of Dodis {\it et al.}~\cite{DGL+99} gave an algorithm for $\eps$-testing the monotonicity of functions $f:[m]^n\to\bits$ with query complexity $O((n/\eps)\log^2(n/\eps))$ (independent of $m$), thus ruling out any dependence on $m$ in a lower bound.}

\medskip

{\bf Our algorithm.}  We present a new algorithm for monotonicity testing and prove the following
result about its performance:

\begin{theorem} \label{main-theorem-alg}
There is a $\tilde{O}(n^{5/6}\eps^{-4})$-query
one-sided non-adaptive algorithm for testing whether an unknown $n$-variable Boolean
function is monotone versus $\eps$-far from monotone.
\end{theorem}
Recall that the one-sided, non-adaptive tester of Chakrabarty and Seshadhri \cite{CS13a}
makes
$\tilde{O}(n^{7/8}\eps^{-3/2})$ queries.  Thus, while the query complexity of our tester
is worse as a function of $1/\eps$ (though still polynomial), its query complexity
is polynomially better as a function of $n$.\footnote{Recall that in property testing the dependence
on the size parameter ``$n$'' is typically viewed as more important than the dependence
on the ``closeness'' parameter $\eps$.  Indeed, $\eps$ is often viewed as a constant, so testers
with query complexities that are exponential (or worse) as a function of $1/\eps$ but independent
of $n$ are commonly referred to as ``constant-query testers.''}
Like the \cite{CS13a} algorithm, our algorithm is a pair tester, but it evades
the $\Omega(n/(\eps \log n))$ lower bound of
\cite{BCGM12} because its query complexity is not of the form $q(n)/\eps$.
Our algorithm builds on the tools developed in \cite{CS13a}; its high-level
structure is similar to that of the \cite{CS13a}
algorithm, but with an important difference that enables an improved
analysis. See Section \ref{sec:ub-approach} for more discussion on this point.

\subsection{The lower bound approach} \label{sec:lb-approach}


Our lower bound for testing monotonicity
builds on previous lower bounds for testing restricted classes of \emph{linear threshold functions} (LTFs).  Recall that $f: \{-1,1\}^n \to \{-1,1\}$ is a
linear threshold function if there exist $w_1,\dots,w_n, \theta \in \R^n$ such that $f(x) = \sign(w \cdot x - \theta)$ for all $x \in \{-1,1\}^n.$

\medskip
{\bf Background.} A \emph{signed majority function} is a linear threshold function of the special form $f(x) = \sign(w \cdot x)$ where $w \in \{-1,1\}^n$.
While \cite{MORS10} showed that the class of all LTFs~is $\eps$-testable using $\poly(1/\eps)$ queries (independent of $n$),
in \cite{MORS09} Matulef \emph{et al.\!} gave an $\Omega(\log n)$ lower bound for non-adaptive algorithms that $\eps_0$-test whether
$f: \{-1,1\}^n \to \{-1,1\}$ is a signed majority function, where $\eps_0>0$ is a universal constant.
{Like many lower bound arguments in property testing,} the proof of \cite{MORS09}
{employs Yao's minimax principle \cite{Yao77}}, and works~by exhibiting two distributions $\calD_{yes}$ and $\calD_{no}$ over LTFs --- more precisely, $\calD_{yes}$ is the uniform
distribution over all $2^n$ signed majority
functions, and $\calD_{no}$ is the uniform distribution over a set of LTFs almost all of which are constant-far from every signed majority function ---
and arguing that for $q=o(\log n)$, any deterministic $q$-query algorithm cannot distinguish between the two distributions with non-negligible success
probability.   (We note that a typical function from $\calD_{yes}$ is far from being monotone, and that the same
holds for a typical LTF drawn from  the $\calD_{no}$ distribution of \cite{MORS09}.)  A key tool in the \cite{MORS09} proof is
the Berry--Ess\'een ``central limit theorem (CLT) with error bounds'' for sums of independent real-valued random variables.

An \emph{embedded majority function of size $k$} is an LTF $f: \{-1,1\}^n \to \{-1,1\}$ of the form $f(x)=\sign(w \cdot x)$ where $w \in \{0,1\}^n$
is a vector with exactly $k$ ones.  In \cite{BO10} Blais and O'Donnell showed that for $k=n/2$, any non-adaptive testing algorithm for the class
of all embedded majority functions of size exactly $n/2$ must make $\Omega(n^{1/12})$ queries.  Their proof employed a $\calD_{yes}$ distribution which is the
uniform distribution over all embedded majority functions of size $n/2$, and a $\calD_{no}$ distribution which is supported on certain monotone
LTFs (which are far from embedded majority functions of size $n/2$).  A key technical ingredient in the proofs of \cite{BO10} is a multidimensional
extension of the Berry--Ess\'een theorem (to independent sums of $\R^q$-valued random variables) which was essentially
established in the work of \cite{GOWZ10}, building on ingredients from \cite{Mos08}.  Subsequently Ron and Servedio \cite{RS13} adapted the arguments of \cite{BO10} to give an improved analysis
of the same $\calD_{yes}$ and $\calD_{no}$ distributions from \cite{MORS09} and establish an
$\Omega(n^{1/12})$-query lower bound for non-adaptive algorithms that $\eps_0$-test whether
$f: \{-1,1\}^n \to \{-1,1\}$ is a signed majority function, thus exponentially improving over the \cite{MORS09} lower bounds for this problem.

\medskip
{\bf This work.}  Neither the \cite{BO10} construction nor the \cite{MORS09,RS13} construction can be used directly to establish a lower bound for
monotonicity testing {of functions $f: \{-1,1\}^n \to \{-1,1\}$}; as described above, in the \cite{BO10} construction both the $\calD_{yes}$ and $\calD_{no}$ functions
are monotone, and in the \cite{MORS09,RS13} construction a typical function from either distribution is far from monotone.  Nevertheless, in this work we
show that ingredients from \cite{BO10,RS13} can be leveraged to obtain a polynomial lower bound for testing
monotonicity {of functions $f: \{-1,1\}^n \to \{-1,1\}$}.  {Like these earlier works we
employ Yao's principle:}  we define a $\calD_{yes}$ distribution that is supported on monotone LTFs, and a $\calD_{no}$ distribution over LTFs that is almost entirely supported
on LTFs that are constant-far from every monotone function, and use an analysis which is fairly similar to that of
\cite{BO10,RS13}, to prove Theorem \ref{main-theorem-lb}.  Using the multidimensional Berry--Ess\'een theorem of
\cite{GOWZ10} to analyze our $\calD_{yes}$ and $\calD_{no}$ distributions would result in an $\Omega(n^{1/12})$ lower bound.  To obtain our improved
$\Omega(n^{1/5} \log^{-2/5} n)$ lower bound, we instead adapt a multidimensional CLT of Valiant and Valiant~\cite{VV11} (for Wasserstein distance) to our context.


\subsection{The approach of our algorithm} \label{sec:ub-approach}

Our algorithm builds on ingredients from \cite{CS13a}, so to
explain our approach we first recall the necessary ingredients from that work.
Fix a Boolean function\footnote{For our algorithmic result
it will be more convenient to view Boolean functions as mapping $\{0,1\}^n$ to $\{0,1\}$.}$f: \{0,1\}^n \to \{0,1\}$, and
let us say that a pair of inputs $(x,y)$ with $x \prec y$ is
a \emph{violated edge} if $f(x)=1,f(y)=0$ and $(x,y)$ is an edge
in $\{0,1\}^n$ (i.e. the Hamming distance between them is 1).
\cite{CS13a} establishes a very useful ``dichotomy theorem'' about Boolean
functions $f: \{0,1\}^n \to \{0,1\}$ that are $\eps$-far from monotone:
for any $s > 0$, any such function either must have $\Omega(\eps s 2^n)$
violated edges, or must have a \emph{matching} (i.e.~a vertex-disjoint set) of $\Omega(\eps 2^n/s)$
violated edges.

To use this dichotomy theorem,
Chakrabarty and Seshadhri \cite{CS13a} define a ``path tester''
which works essentially as follows:  it selects a random
directed path $\bp$ of $n$ edges from $0^n$ up to $1^n$, draws two uniform
random points $\bx \prec \by$ from the ``middle layers'' of $\bp$,
and rejects if $\bx$ and $\by$ violate monotonicity, i.e.
$f(\bx)=1$ and $f(\by)=0$.\footnote{Here the ``middle layers'' of $\bp$
are the points on the path that have $n/2 \pm O_\eps(\sqrt{n})$ many
coordinates which are $1$; intuitively, at most an $\eps$-fraction
of all points in $\{0,1\}^n$ lie outside these ``middle layers''
of the hypercube.  We note that the above description is
 a slight simplification of the actual \cite{CS13a}
path tester, omitting some
details which are not necessary at this stage of our description.}  They prove
that if $f$ has a matching of $\Omega(\sigma 2^n)$ violated
edges, then their path tester will uncover a violation and reject
with probability $\tilde{\Omega}(\sigma^3/\sqrt{n})$.  (Roughly speaking, they show that about an $\Omega(\sigma)$ fraction of possible outcomes of $\by$, corresponding to the $\sigma 2^n$ upper endpoints of the edges in the matching, are such that with probability $\tilde{\Omega}(\sigma^2/\sqrt{n})$ over the random draw of $\bx$, the pair
$\by$ and $\bx$ together constitute a violation.) On
the other hand, if $f$ does not have a matching of this size then
(by the dichotomy theorem) it must have $\Omega((\eps^2/\sigma)2^n)$ violated
edges, so the edge tester of \cite{GGL+98} (querying the endpoints of a
uniform random edge) will hit a violated edge with probability
$\Omega(\eps^2/(\sigma n))$.  Their final algorithm runs their path tester
with probability $1/2$ and queries a random edge with probability $1/2$.
Choosing $\sigma$ suitably to equalize the two rejection probabilities,
this is a two-query algorithm which succeeds in uncovering a violation
for any $\eps$-far-from-monotone function $f$ with probability
$\tilde{\Omega}(\eps^{3/2}/n^{7/8})$, giving them a one-sided non-adaptive
tester which makes $\tilde{O}(n^{7/8}/\eps^{3/2})$ queries overall.

Our algorithm follows the same high-level framework described above,
but differs from \cite{CS13a} by employing a different path tester.
After selecting a random path $\bp$, instead of (essentially) drawing two
independent uniform points from the middle layers of the path as is
done in \cite{CS13a},
our path tester draws a \emph{correlated} pair of points from $\bp$.
More precisely, it selects the first point $\by$ independently
from the middle layers of $\bp$,
and preferentially selects the second point $\bx$ from
$\bp$ in a way which favors points which are closer to $\by$.  Via a careful
analysis we are able to show that if $f$ has a matching of
$\Omega(\sigma 2^n)$ violated edges, then our
path tester will uncover a violation and reject
with probability $\tilde{\Omega}(\sigma^2/\sqrt{n})\cdot \poly(\eps)$.
Roughly speaking, we show that if $\by$ is a uniform random upper endpoint of the $\sigma 2^n$
edges in the matching (which occurs with probability about $\sigma$), then the probability that our tester selects a  string $\bx$ which gives a violation with $\by$ is $\tilde{\Omega}(\sigma/\sqrt{n})\cdot\poly(\eps).$ Trading this off against the success probability of the edge tester
using the dichotomy theorem, we obtain our improved query bound.

\medskip
{\bf Organization of this paper.}
Our lower bound results are established  Sections \ref{two-dist} through
\ref{sec:hypergrid}.
The two distributions $\calD_{yes}$ and $\calD_{no}$ are defined at the beginning of Section~\ref{two-dist}. In Section~\ref{sec:far-from-monotone} we show that with high probability an LTF drawn from $\calD_{no}$ is constant-far from monotone, and in Section~\ref{sec:cant-distinguish} we show that unless
$q=\Omega(n^{1/5} (\log n)^{-2/5})$, any deterministic $q$-query algorithm cannot distinguish between the two distributions with non-negligible success probability. The key technical ingredient in our proof of the latter is a lemma that adapts the Valiant--Valiant multidimensional CLT for Wasserstein distance to our context; we prove this lemma in Section~\ref{sec:vv}.
Finally in Section~\ref{sec:hypergrid} we prove Theorem~\ref{main-hypergrid}, showing that the same lower bound of $\tilde{\Omega}(n^{1/5})$ also applies to the query complexity of testers for monotonicity of functions $f:[m]^n\to\zo$ over general hypergrid domains; we do so via a reduction to the $m=2$ case (Theorem~\ref{main-theorem-lb}).

Our algorithmic result is established in Section \ref{sec-alg}.
In Section~\ref{sec-two-useful-dist} we describe two useful distributions over comparable pairs $(\bx,\by)$
  from the middle layers of $\{0,1\}^n$ and bound the probability of having both points landing
  in a fixed set $A$ of size $\sigma 2^n$.
Then in Section~\ref{sec-alg-score} we define the \emph{score} of a point $x$ with respect to a set
  $A$ of points, and use the result of Section~\ref{sec-two-useful-dist} to lower bound the sum of
  $\score(x,A)$ over all points $x\in A$.
We present our modified path tester as well as the analysis of its success probability
  in Section~\ref{our-path-tester}.
Finally in Section~\ref{sec-alg-proof} we combine this tester and the dichotomy theorem of \cite{CS13a}
  to obtain our improved upper bound.

\subsection{Preliminaries}

All probabilities and expectations are with respect to the uniform distribution unless otherwise stated; we will use boldface letters (e.g. $\bx$ and $\bX$) to denote random variables.    For a $q\times n$ matrix $Q \in \R^{q\times n}$, we write $Q_{i*} \in \R^n$ to denote its $i$-th row, $Q_{*j} \in \R^q$ its $j$-th column, and $Q_{i,j} \in \R$ its entry in the $i$-th column and $j$-th row.  We write $\prec$ to denote the coordinate-wise partial order on $\bn$, where $x \prec y$ iff $x_i \le y_i$ for all $i\in [n]$ and $x\ne y$.  We say that $x$ and $y$ are $\emph{comparable}$ if $x \prec y$, $y \prec x$, or $x=y$.  Given two functions $f,g:\bn\to\bits$ we will use $\dist(f,g)$ to denote the (normalized Hamming) distance $\Prx_{\bx\in\bn}[f(\bx)\ne g(\bx)]$ between $f$ and $g$.\vspace{0.035cm}

Recall that $f:\bn\to\bits$ is monotone if $f(x) \le f(y)$ for all $x,y \in \bn$ such that $x\prec y$. We say that $f$ is \emph{$\eps$-close to monotone} if $\dist(f,g) \le \eps$ for some monotone $g:\bn\to \bits$, and \emph{$\eps$-far from monotone} otherwise. A linear threshold function (LTF) over $\bn$~is a function $f:\bn\to{\bits}$ that can be expressed as $f(x) = \sign(w\cdot x - \theta)$ for some $w_1,\ldots,w_n$, $\theta\in \R$.  Here $\sign : \R \to\bits$ is the sign function $\sign(t) = 1$ if $t\ge 0$ and $\sign(t) = -1$ if $t < 0$.  {For $f(x) = \sign(w \cdot x - \theta)$,
an LTF over $\bn$, it is straightforward to verify that if $w_i \ge 0$ for all $i \in [n]$ then
$f$ is monotone.}\vspace{0.02cm}

We will need a few standard facts from probability theory:

\begin{fact}[Gaussian anti-concentration]
\label{gaussian-anti-concentration}
Let $\calG$ be a Gaussian with variance $\sigma^2$. Then for all $\eps > 0$
it holds that
$\sup_{\theta \in \R}\big\{ \Pr\big[|\calG-\theta| \le \eps\sigma \big]\big\} \le \eps.$
\end{fact}

\begin{fact}[Gaussian concentration]
\label{gaussian-concentration}
Let $\calG$ be a Gaussian with mean {0} and variance $\sigma^2$. Then for all {$0<a<1$
it holds that $\Pr\big[\calG \in [0,a\sigma] \big]  = \Omega(a).$}
\end{fact}

\begin{theorem}[Berry--Ess\'een]
\label{be}
Let $\bS = \bX_1 + \cdots + \bX_n$ where $\bX_1,\ldots,\bX_n$ are independent real-valued random variables with $\E[\bX_j] = \mu_j$ and $\Var[\bX_j] = \sigma_j^2$, and suppose that $| \bX_j - \E[\bX_j]|\le \tau$ with pro\-bability $1$ for all $j\in [n]$. Let $\calG$ be a Gaussian with mean $\sum_{j=1}^n \mu_j$ and variance $\sum_{j=1}^n\sigma_j^2$, matching those of $\bS$. Then for all $\theta\in \R$, we have
\[ \big| \Pr[\bS \le \theta] - \Pr[\calG \le \theta] \big| \le \frac{O(\tau)}{\big(\sum_{j=1}^n \sigma_j^2\big)^{1/2}}.\]
\end{theorem}

\begin{fact}
\label{binomial-weight}
For all $c > 0$ there exists an $\eps = \eps(c) \in (0,1]$ such that the following holds. For all even (resp.\ odd) $n$,
\[ \Prx_{\bx \in\bn}\Bigg[\sumi\bx_i = k\Bigg] \ge \frac{\eps}{\sqrt{n}}\ \text{ for all even (resp.\ odd) integers $k \in [-c\sqrt{n},c\sqrt{n}]$}. \]
\end{fact}

Finally we recall a few basic facts from the Fourier analysis over the hypercube which we require (for a comprehensive treatment of this topic see~\cite{odonnell}).  Every function $f:\bn\to\R$ can be uniquely expressed as a multilinear polynomial
\[ f(x) = \sumS \hatf(S) \prod_{i\in S}\orange{x_i}\quad \text{where $\hatf(S) := \Ex_{\bx\in\bn}\bigg[f(\bx)\prod_{i\in S} \bx_i\bigg]$,} \]
known as the \emph{Fourier transform} of $f$.  The numbers $\hatf(S) \in \R$ are the \emph{Fourier coefficients} of $f$; with a slight abuse of notation we will write $\hat{f}(i)$ instead of $\hatf(\{i\})$ for the degree-$1$ Fourier coefficients.
\begin{fact}[Parseval's identity]
Let $f:\bn\to\R$. Then
\[ \Ex_{\bx\in\bn}[f(\bx)] = \sumS \hatf(S)^2. \]
\end{fact}
For $i \in [n]$, the \emph{influence of coordinate $i$ on $f$}, denoted $\Inf_i[f]$, is the probability
\[ \Inf_i[f] := \Prx_{\bx \in\bn}\big[f(\bx)\ne f(\bx^{\oplus i})\big], \]
where $\bx^{\oplus i}$ denotes the string $\bx$ with its $i$-th coordinate flipped.   The following fact relates the influences of an LTF to its degree-$1$ Fourier coefficients:

 \begin{fact}
\label{linear-coeff-influence}
Let $f(x) = \sign(w_1x_1 + \cdots + w_n x_n-\theta)$ be an LTF over $\bn$.  Then for all $i\in [n]$, $\Inf_i[f] = \hat{f}(i)$ if $w_i \ge 0$ and $\Inf_i[f] = -\hat{f}(i)$ if $w_i < 0$.
\end{fact}

\section{The lower bound:  Proof of Theorem~\ref{main-hypergrid}}
\label{two-dist}
Let $\calD_{yes}$ be the following distribution over monotone LTFs {on $\bn$}: a draw $\boldf_{yes} \sim\calD_{yes}$ is $\boldf_{yes}(x) = \sign(\bsigma_1 x_1 + \cdots + \bsigma_n x_n)$, where each $\bsigma_i$ is independently and uniformly chosen from $\{1,3\}$.  The distribution $\calD_{no}$ is similarly a distribution over LTFs $\boldf_{no}(x) = \sign(\bnu_1x_1 + \cdots + \bnu_nx_n)$, but each $\bnu_i$ is independently chosen to be $-1$ with probability $1/10$, and $7/3$ with probability $9/10$.  The following two propositions along with a standard application of Yao's minimax principle~\cite{Yao77} yield Theorem~\ref{main-hypergrid}:

\begin{proposition}
\label{far-from-monotone}
There exists a universal positive constant $\eps_0 > 0$ such that with probability $1-o_n(1)$, a random LTF $\boldf_{no} \sim \calD_{no}$ satisfies $\dist(\boldf_{no},g) > \eps_0$ for all monotone Boolean functions $g:\bn\to\bits$.
\end{proposition}

\begin{proposition}
\label{cant-distinguish}
Let $\calT$ be any deterministic non-adaptive two-sided $q$-query algorithm for testing whether a black-box Boolean function $f:\bn\to\bits$ is monotone. Then
\begin{equation} \bigg| \Prx_{\boldf_{yes}\sim\calD_{yes}}\Big[\text{$\calT$ {outputs {\sf Yes} on} $\boldf_{yes}$}\Big] -
\Prx_{\boldf_{no}\sim\calD_{no}}\Big[\text{$\calT$ {outputs {\sf Yes} on} $\boldf_{no}$}\Big]\bigg|  =
O\bigg(\frac{q^{5/4}(\log n)^{1/2}}{n^{1/4}}\bigg).
\label{eq:cant-distinguish} \end{equation}
\end{proposition}

We prove Proposition~\ref{far-from-monotone} in Section~\ref{sec:far-from-monotone}, followed by Proposition~\ref{cant-distinguish} in Section~\ref{sec:cant-distinguish}.

\subsection{Proof of Proposition~\ref{far-from-monotone}}
\label{sec:far-from-monotone}
 By the Chernoff bound, with probability $1-o_n(1)$ a draw $\boldf_{no} = \sign(\bnu_1x_1 + \cdots + \bnu_nx_n)$ from $\calD_{no}$ satisfies
\[ | \{i\in [n] \colon \bnu_i = -1\} | \in \Big[0.1n -\sqrt{n\log n}, 0.1n+\sqrt{n\log n}\Big].\]
We call any such LTF \emph{nice}, and we will argue that all nice LTFs are constant-far from monotonicity. For the remainder of this proof let $f$ be a nice LTF, which we may without loss of generality express as $f(x) = \sign(\ell(x))$ where
\[ \ell(x) := -(x_1 + \cdots + x_m) + \lfrac{7}{3} \cdot (x_{m+1} + \cdots + x_n)\]
and $m\in \big[0.1n -\sqrt{n\log n}, 0.1n+\sqrt{n\log n}\big]$.  We assume that $m$ is odd, noting that the case when $m$ is even follows via an identical argument.  We first claim that $\Inf_i[f] = \Omega(1/\sqrt{n})$ for all $i\in [m]$; by symmetry it suffices to show this for $i=1$.  Define $\ell'(x) := - (x_2 + \ldots + x_m) + \frac{7}{3}  (x_{m+1} + \cdots + x_n)$ and note that $f(x) \ne f(x^{\oplus 1})$ if and only if $\ell'(x) \in [-1,1)$.  Applying Fact \ref{binomial-weight} twice, we have
\[ \Prx_{\bx\in\bn}\Big[\lfrac{7}{3}(\bx_{m+1} + \cdots + \bx_n) \in [-\sqrt{n},\sqrt{n}]\Big] = \Omega(1) \]
and
\[ \Prx_{\bx\in\bn}[\bx_2+\cdots + \bx_m = k] = \Omega(1/\sqrt{n}) \text{ for all even integers $k \in [-\sqrt{n}-1,\sqrt{n}+1]$}, \]
and therefore indeed,
\[ \Inf_1[f] = \Prx_{\bx \in\bn}[\ell'(\bx) \in [-1,1)] = \Omega(1/\sqrt{n}).\]
Since $\Inf_i[f] = \Omega(1/\sqrt{n})$ for all $i\in [m]$, by Fact \ref{linear-coeff-influence} we have that $\hat{f}(i) = -\Omega(1/\sqrt{n})$ for all $i\in [m]$. Hence for all monotone Boolean functions $g$, we have
\begin{eqnarray*} 4\cdot \dist(f,g)\ =\ \Ex_{\bx\in\bn}\big[(f(\bx)-g(\bx))^2\big] &=& \sum_{S\sse [n]} \big(\hatf(S)-\hat{g}(S)\big)^2 \\
&\ge& \sum_{i=1}^m \big(\hatf(i)-\hat{g}(i)\big)^2 \ = \ m\cdot \Omega(1/n)\ =\ \Omega(1).
\end{eqnarray*}
Here the second equality is by Parvseval's identity; the penultimate equality uses the fact that $\hat{g}(i) \ge 0$ for all $i\in [n]$, which in turn holds since $g$ is a monotone Boolean function.  This completes the proof of Proposition~\ref{far-from-monotone}.

 \subsection{Proof of Proposition \ref{cant-distinguish}}
\label{sec:cant-distinguish}
Let $\calT$ be any deterministic non-adaptive $q$-query tester, and view its $q$ queries as a $q\times n$ matrix $Q \in \bits^{q\times n}$.  Following the terminology of~\cite{BO10}, we define a ``Response Vector'' random variable $\bR_{yes} \in \bits^q$ which is obtained by drawing $\boldf_{yes}= \sign(\bsigma_1 x_1 + \cdots + \bsigma_n x_n)$ from $\calD_{yes}$ and setting the $i$-th coordinate of $\bR_{yes}$ to be \[ \boldf_{yes}(Q_{i*}) = \sign(\bsigma_1 Q_{i,1} + \cdots + \bsigma_n Q_{i,n}),\]
and similarly $\bR_{no} \in \bits^q$ which is obtained by drawing $\boldf_{no}\sim \calD_{no}$ and setting the $i$-th coordinate of $\bR_{no}$ to be $\boldf_{no}(Q_{i*})$.  By the definition of total variation distance, the left-hand side of (\ref{eq:cant-distinguish}) is upper bounded by $\dtv(\bR_{yes},\bR_{no})$, and hence we can prove Proposition \ref{cant-distinguish} by showing that $\dtv(\bR_{yes},\bR_{no}) = O(q^{5/4}(\log n)^{1/2}/n^{1/4})$.

Let $\bS \in \R^q$ be the random column vector $Q\bsigma$ where $\bsigma$ is uniform over $\{1,3\}^n$, and $\bT\in\R^q$ be the random column vector $Q\bnu$ where $\bnu$ is drawn from the product distribution over $\{-1,7/3\}^n$ where $\Pr[\bnu_i = -1] =1/10$ for all $i\in [n]$.  The Response Vector $\bR_{yes}$ is determined by the orthant of $\R^q$ in which $\bS$ lies (as each coordinate of $\bR_{yes}$ is simply the sign of the respective coordinate of $\bS$), and likewise $\bR_{no}$ by the orthant of $\R^q$ in which $\bT$ lies.  Therefore it suffices for us to prove the following lemma:

\begin{lemma}
\label{S-T}
Let $\bS, \bT\in \R^q$ be defined as above. Then for any union $\calO$ of orthants in $\R^q$,
\[ \big|\Pr[\bS \in \calO] - \Pr[\bT\in\calO] \big| = O\bigg(\frac{q^{5/4}(\log n)^{1/2}}{n^{1/4}}\bigg).\]
\end{lemma}

We will need the following multidimensional Berry--Ess\'een theorem, the proof of which we defer to Section~\ref{sec:vv}.
\begin{theorem}
\label{vv-corollary}
Let $\bS = \bX^{(1)} + \cdots + \bX^{(n)}$ where $\bX^{(1)},\ldots,\bX^{(n)}$ are independent $\R^q$-valued random variables, and suppose that $\big| \bX^{(j)}_i - \E[\bX^{(j)}_i]\big|\le \tau$ with probability $1$ for all $i\in [q]$ and $j\in [n]$. Let $\calG$ be the $q$-dimensional Gaussian with the same mean and covariance matrix as $\bS$.  Let $\calO$ be a union of orthants in $\R^q$. Then for all $r > 0$,
\[ \big|\Pr[\bS \in \calO] - \Pr[\calG\in\calO] \big| = O\Bigg( \frac{\tau q^{3/2} \log n}{r} + \sum_{i=1}^q \frac{r+\tau}{\big(\sum_{j=1}^n \Var\big[\bX^{(j)}_i\big]\big)^{1/2}}\Bigg).\]
\end{theorem}

\begin{proof}[Proof of Lemma \ref{S-T} assuming Theorem~\ref{vv-corollary}]
We begin by writing $\bS = \bX^{(1)} + \cdots + \bX^{(n)}$, where $\bX^{(j)} = \bsigma_j \cdot Q_{*j} $ and $\bsigma_j$ is uniform over $\{1,3\}$; {\it i.e.}~each $\bX^{(j)}$ is independently $Q_{*j}$ with probability $1/2$ and $3\cdot Q_{*j}$ with probability $1/2$. Likewise we may express $\bT = \bY^{(1)} + \cdots + \bY^{(n)}$, where $\bY^{(j)} = \bnu_{j}\cdot Q_{*j}$ and $\bnu_j$ is $-1$ with probability $1/10$ and $7/3$ with probability $9/10$.  We claim that the $\bX^{(j)}$'s and $\bY^{(j)}$'s have matching means and covariance matrices; it suffices to check this for $\bX^{(1)}$ and $\bY^{(1)}$.  For means, we see that indeed
\begin{align*}
\E\big[\bX^{(1)}\big] &= \E[\bsigma_1] \cdot Q_{*1} = (\lfrac{1}{2} + \lfrac{3}{2}) \cdot Q_{*1} = 2\cdot  Q_{*1} \\
\E\big[\bY^{(1)}\big] &= \E[\bnu_1] \cdot Q_{*1} = (-\lfrac{1}{10} + \lfrac{9}{10}\lfrac{7}{3}) \cdot Q_{*1} = 2\cdot Q_{*1}.
\end{align*}
As for the covariance matrices, we let $i_1,i_2 \in [q]$ and calculate
\begin{eqnarray*} \Cov[\bX^{(1)}]_{i_1,i_2} &=& \E\big[\big(\bX^{(1)}_{i_1} - 2\cdot Q_{i_1,1}\big)\big(\bX^{(1)}_{i_2} - 2Q_{i_2,1}\big)\big] \\
&=& \E\big[\bX^{(1)}_{i_1} \cdot \bX^{(1)}_{i_2}\big] - 2\cdot Q_{i_2,1}\E\big[\bX^{(1)}_{i_1}\big] - 2\cdot Q_{i_1,1}\E\big[\bX^{(1)}_{i_2}\big] + 4\cdot Q_{i_1,1}Q_{i_2,1} \\
&=&  \E\big[\bX^{(1)}_{i_1} \cdot \bX^{(1)}_{i_2}\big] -4\cdot Q_{i_1,1}Q_{i_2,1}\\
&=& \left(\E\big[\bsigma_1^2\big] - 4\right)\cdot Q_{i_1,1}Q_{i_2,1}  \ = \ (\lfrac{1}{2} + \lfrac{9}{2} -4) \cdot Q_{i_1,1}Q_{i_2,1} \ =  Q_{i_1,1}Q_{i_2,1}.
\end{eqnarray*}
Similarly, the corresponding entry of $\Cov[\bY^{(1)}]$ is:
\begin{eqnarray*} \Cov[\bY^{(1)}]_{i_1,i_2} &=& \E\big[\big(\bY^{(1)}_{i_1} - 2\cdot Q_{i_1,1}\big)\big(\bY^{(1)}_{i_2} - 2Q_{i_2,1}\big)\big] \\
&=& \E\big[\bY^{(1)}_{i_1} \cdot \bY^{(1)}_{i_2}\big] - 2\cdot Q_{i_2,1}\E\big[\bY^{(1)}_{i_1}\big] - 2\cdot Q_{i_1,1}\E\big[\bY^{(1)}_{i_2}\big] + 4\cdot Q_{i_1,1}Q_{i_2,1} \\
&=&  \E\big[\bY^{(1)}_{i_1} \cdot \bY^{(1)}_{i_2}\big] -4\cdot Q_{i_1,1}Q_{i_2,1}\\
&=& \left(\E\big[\bnu_1^2\big] - 4\right)\cdot  Q_{i_1,1}Q_{i_2,1}\ = \ (\lfrac{1}{10}+\lfrac{9}{10}\lfrac{49}{9} -4) \cdot Q_{i_1,1}Q_{i_2,1} \ =
Q_{i_1,1}Q_{i_2,1}.
\end{eqnarray*}
Since the $\bX^{(j)}$'s and $\bY^{(j)}$'s have matching means and covariance matrices, so do their sums $\bS$ and $\bT$, and so Theorem~\ref{vv-corollary} gives a bound on the differences $|\Pr[\bS\in\calO]-\Pr[\calG\in\calO]|$ and $|\Pr[\bT\in\calO]-\Pr[\calG\in\calO]|$ for the same $q$-dimensional Gaussian $\calG$.  Recalling that
$\bX^{(j)}_i = \bsigma_j \cdot Q_{i,j}$ where $Q_{i,j}\in\bn$, we have that $\Var[\bX_i^{(j)}] = 1$, and likewise $\Var[\bY^{(j)}_i] = 1$.
Therefore, two applications of Theorem~\ref{vv-corollary} with $\tau :=  O(1)$
along with the triangle inequality yields the bound
\[ \big| \Pr[\calS\in\calO]-\Pr[\calG\in\calO]\big| = O\Bigg( \frac{q^{3/2} \log n}{r} + \frac{q(r+\tau)}{\sqrt{n}}\Bigg)    \]
for all $r  > 0$.
Choosing $r := (qn)^{1/4}(\log n)^{1/2}$ completes the proof.
\end{proof}

\section{Multidimensional Berry--Ess\'een via the Valiant--Valiant CLT}
\label{sec:vv}
In this section we prove Theorem \ref{vv-corollary} by adapting a recent multidimensional CLT of Valiant and Valiant~\cite{VV11}
which bounds the \emph{Wasserstein distance} between a sum of independent
vector-valued random variables and a multidimensional Gaussian.

\begin{definition}[Wasserstein distance]
The Wasserstein distance between two $\R^q$-valued random variables $\bS$ and $\bT$, denoted $d_W(\bS,\bT)$, is defined to be:
\[ d_W(\bS,\bT) = \inf_{\calD}\Big\{ \Ex_{\calD}\big[\|\bU-\bV\|_2\big] \Big\}, \]
where the infimum is taken over {all} couplings $\calD$ of $\bS$ and $\bT$, {i.e.,}  all joint distributions $\calD$ of pairs of
$\R^q$-valued random variables ($\bU,\bV)$ with marginals distributed according to $\bS$ and $\bT$ respectively.
\end{definition}

Valiant and Valiant \cite{VV11} recently used Stein's method to prove the following central limit theorem for Wasserstein distance:

\begin{theorem}[Valiant--Valiant CLT]
\label{vv}
Let $\bS = \bX^{(1)} + \cdots + \bX^{(n)}$ where $\bX^{(1)},\ldots,\bX^{(n)}$ are independent $\R^q$-valued random variables, and suppose $\big\| \bX^{(j)} - \E\big[\bX^{(j)}\big] \big\|_2 \le \beta$ with probability $1$ for any $j\in [n]$. Then
\[ d_W(\bS , \calG) \le O(\beta q \log n),  \]
where $\calG$ is the $q$-dimensional Gaussian with the same mean and covariance matrix as $\bS$.
\end{theorem}


We recall Theorem \ref{vv-corollary}:

\begin{reptheorem}{vv-corollary}
Let $\bS = \bX^{(1)} + \cdots + \bX^{(n)}$ where $\bX^{(1)},\ldots,\bX^{(n)}$ are independent $\R^q$-valued random variables, and suppose that $\big| \bX^{(j)}_i - \E\big[\bX^{(j)}_i\big]\big|\le \tau$ with probability $1$ for all $i\in [q]$ and $j\in [n]$. Let $\calG$ be the $q$-dimensional Gaussian with the same mean and covariance matrix as $\bS$.  Let $\calO$ be a union of orthants in $\R^q$. Then for all $r > 0$,
\[ \big|\Pr[\bS \in \calO] - \Pr[\calG\in\calO] \big| = O\Bigg( \frac{\tau q^{3/2} \log n}{r} + \sum_{i=1}^q \frac{r+\tau}{\big(\sum_{j=1}^n \Var\big[\bX^{(j)}_i\big]\big)^{1/2}}\Bigg).\]
\end{reptheorem}

\begin{proof}
 We define \[ W_r := \big\{ x \in \R^q \colon |x_i| \le r \text{ for some $i\in [q]$}\big\} \]
to be the radius-$r$ region around the orthant boundaries, and  partition $\calO$ into $\calO_{bd} := \calO \cap W_r$ (the points in $\calO$ that lie close to the orthant boundaries) and $\calO_{in} := \calO \setminus W_r$ (the points that lie far away from the orthant boundaries). We have
\begin{eqnarray}
 \big|\Pr[\bS \in \calO] - \Pr[\calG\in\calO] \big| &=&  \big|(\Pr[\bS \in \calO_{in}] + \Pr[\bS \in\calO_{bd}]) - (\Pr[\calG\in\calO_{in}]  + \Pr[\calG\in\calO_{bd}]) \big| \nonumber \\
 & \le& \underbrace{\big|\Pr[\bS \in \calO_{in}] - \Pr[\calG \in\calO_{in}]\big|}_{\Delta} + \underbrace{\Pr[\bS\in\calO_{bd}]  + \Pr[\calG\in\calO_{bd}]}_{\Gamma}. \label{eq:S-N1}
\end{eqnarray}
We bound the quantities $\Delta$ and $\Gamma$ separately. For $\Gamma$, we have that
\begin{eqnarray} \Gamma &\le& \sum_{i=1}^q \Pr\big[\bS_i \in [-r,r]\big] + \Pr\big[\calG_i \in [-r,r]\big] \label{eq:S-N2} \\
&\le  & \sum_{i=1}^q 2\Pr\big[\calG_i \in [-r,r]\big]  + \big| \Pr\big[\bS_i \in [-r,r]\big] - \Pr\big[\calG_i \in [-r,r]\big]\big| \nonumber \\
&\le &\sum_{i=1}^q  \frac{O(r)}{\big(\sum_{j=1}^n \Var\big[\bX^{(j)}_i\big]\big)^{1/2}} + \frac{O(\tau)}{\big(\sum_{j=1}^n \Var\big[\bX^{(j)}_i\big]\big)^{1/2}} \ = \ \sum_{i=1}^q  \frac{O(r+\tau)}{\big(\sum_{j=1}^n \Var\big[\bX^{(j)}_i\big]\big)^{1/2}} \label{eq:S-N3}
\end{eqnarray}
where (\ref{eq:S-N2}) is a union bound over all $q$ dimensions, and (\ref{eq:S-N3}) uses  Fact~\ref{gaussian-anti-concentration} (Gaussian anti-concentration),
the fact that $\calG_i$ is a Gaussian with variance $\sum_{j=1}^n \Var\big[\bX^{(j)}_i\big]$,  and Theorem~\ref{be} (Berry--Ess\'een).

For $\Delta$, let us assume without loss of generality (a symmetrical argument works in the other case) that $\Pr[\bS \in \calO_{in}] \ge \Pr[\calG \in\calO_{in}]$, so $\Delta =
\Pr[\bS \in \calO_{in}] - \Pr[\calG \in\calO_{in}].$
 Let $\calD$ be any coupling of $\bS$ and {$\calG$}, so $\calD$ is the
joint distribution of a pair $(\bU,\bV)$ of $\R^q$-valued random variables with marginals distributed according to $\bS$ and {$\calG$} respectively.  Since
\[\int_{\calO_{in}} \int_{\R^q} \calD(u,v)\blue{\,dv\,du} = \Pr[\bS \in \calO_{in}]  \]
and
\[ \int_{\calO_{in}} \int_{\calO_{in}} \calD(u,v)\blue{\,dv\,du} \le \int_{\R^q} \int_{\calO_{in}} \calD(u,v)\blue{\,dv\,du} = \Pr[\calG \in \calO_{in}],\]
it follows that
\begin{equation} \int_{\calO_{in}} \int_{\R^q \setminus \calO_{in}} \calD(u,v)\blue{\,dv\,du} = \int_{\calO_{in}} \int_{\R^q} \calD(u,v)\blue{\,dv\,du} - \int_{\calO_{in}} \int_{\calO_{in}} \calD(u,v)\blue{\,dv\,du} \ge \Delta. \label{eq:emd1} \end{equation}
Next we define the quantities
\newcommand{\Deltanear}{\Delta_{\text{\it near}}}
\newcommand{\Deltafar}{\Delta_{\text{\it far}}}
\begin{eqnarray*}
\Delta_{\text{\it near}}(\calD) &:=&   \int_{\calO_{in}} \int_{\calO_{bd}} \calD(u,v)\blue{\,dv\,du} \\
\Delta_{\text{\it far}}(\calD) &:=&  \int_{\calO_{in}} \int_{\R^q \setminus \calO} \calD(u,v)\blue{\,dv\,du}. \\
\end{eqnarray*}
Note that $\Deltanear(\calD)$ and $\Deltafar(\calD)$ sum to the quantity on the left-hand side of (\ref{eq:emd1}), and so $\Deltanear(\calD)+\Deltafar(\calD) \ge \Delta$.  (In words, since $\bS$ places $\Delta$ more mass on $\calO_{in}$ than $\calG$ does, any scheme $\calD$ of moving the mass of $\bS$ to obtain $\calG$ must move at least $\Delta$ amount from within $\calO_{in}$ to outside it.  $\Deltanear(\calD)$ {is} the amount moved from within $\calO_{in}$ to $\calO$'s boundary $\calO_{bd}$, and $\Deltafar(\calD)$ {is} the rest, moved from within $\calO_{in}$ to locations entirely out of $\calO$.)  Since $\| u - v \|_2 \ge r$ for any pair of points $u\in \calO_{in}$ and $y \notin \calO$, it follows that
\[ d_W(\bS,\calG) \ge r \cdot \Deltafar(\calD). \]
We consider two cases, depending on the relative magnitudes of $\Deltanear(\calD)$ and $\Deltafar(\calD)$.  If $\Deltafar(\calD) \ge \Deltanear(\calD)$, we first observe that for all $j\in [n]$ we have $\big\| \bX^{(j)}-\E\big[\bX^{(j)}\big]\big\|_2 \le \tau \sqrt{q}$ with probability $1$, since each of its $q$ coordinates $i$ satisfies $\big|\bX^{(j)}_i - \E\big[\bX^{(j)}_i\big]\big| \le \tau$ with probability $1$ by the assumption of the theorem. Therefore  we may apply Theorem~\ref{vv} (Valiant--Valiant CLT) with $\beta := \tau\sqrt{q}$ to get
\[ r\cdot \frac{\Delta}{2} \le r\cdot \Deltafar(\calD) \le d_W(\bS,\calG) = O(\tau q^{3/2}\log n) \]
and hence $\Delta = O((\tau q^{3/2}\log n)/r)$, which along with our upper bound on $\Gamma$ completes the proof.  If on the other hand $\Deltanear(\calD) > \Deltafar(\calD)$, then
\[  \frac{\Delta}{2} \le \Deltanear(\calD) \le \int_{\R^q} \int_{\calO_{bd}} \calD(u,v)\blue{\,dv\,du} =  \Pr[\calG \in \calO_{bd}] \le \Gamma, \]
and again our bound on $\Gamma$ completes the proof.
 \end{proof}

\section{A lower bound for general hypergrid domains}
\label{sec:hypergrid}
In this section we prove Theorem~\ref{main-hypergrid}, showing that for all  $m\in \N$ essentially
the same lower bound of
$\tilde{\Omega}(n^{1/5})$ also applies to the query complexity of testers for monotonicity of functions $f:[m]^n\to\bits$, Boolean-valued functions over general hypergrid domains. The notions of monotonicity and distance to monotonicity of functions generalize to functions $f:[m]^n\to\bits$ the natural way: $f$ is monotone if $f(x) \le f(y)$ for all $x \prec y$, where $x \prec y$ iff $x_i \le y_i$ for all $i\in [n]$ and $x\ne y$.  We say that $f$ is $\eps$-close to monotone if $\Pr_{x\in [m]^n}[f(x)\ne g(x)]\le \eps$ for some monotone $g:[m]^n\to\bits$, and $\eps$-far from monotone otherwise.

We prove Theorem~\ref{main-hypergrid} via a reduction to the $m=2$ case (i.e.~Theorem~\ref{main-theorem-lb}).  The reduction is simpler for even $m$ so for ease
of exposition we assume below that $m$ is even.  In this case Theorem~\ref{main-hypergrid} is a direct consequence of Theorem~\ref{main-theorem-lb} and the following proposition:

\begin{proposition}
\label{distance-preserving}
For all even $m\in \N$ the mapping
\[ \Phi : \big\{ \text{all functions $f:\{-1,1\}^n\to\bits$} \big\} \to \big\{  \text{all functions $f : [m]^n\to\bits$}\big\} \]
defined by \emph{(\ref{eq:Phim})} below\ignore{\rnote{Bit of a pedantic point here:  I stuck this in because I think a priori just knowing there exists some mapping might not be enough.  Suppose that the mapping were totally bizarre so that
in order to compute the value of $\Phi[f](z)$ for a given $z$ you needed to, like, query $f$ on many many points $x \neq z$.  Then I don't think the reduction would work.  We're not saying it, but the idea is that if you had an efficient tester for the $[m]^n$ domain, using this $\Phi$ you could convert MQ access to a Boolean $f$ to MQ access to $\Phi[f]$ and thereby test the original Boolean function.  If simulating MQ access to $\Phi[f]$ were very expensive though then this wouldn't work.  Our thing works because given any $x \in [m]^n$ on which we want to query $\Phi[f]$, we can do this (using an MQ oracle for $f$) with just one query.

Let me know if this makes sense?  Even if not it's a small change so
hopefully ok.}} satisfies the following two properties:
\begin{enumerate}
\item If $f\isafunc$ is monotone then $\Phi[f]$ is monotone as well.
\item If $f\isafunc$ is $\eps$-far from monotone then $\Phi[f]$ is $\eps$-far from monotone as well.
\end{enumerate}
\end{proposition}

We will need the following characterization of distance to monotonicity.

\begin{theorem}[\cite{FLN+02} Lemma 4]
\label{violating-pairs}
\ignore{ \rnote{Is this characterization really in \cite{DGL+99}, where is it,
Lemma 7?  That seems to be for general ranges and has a factor-of-2 loss
if I'm reading it correctly.  The \cite{FLN+02} paper seems to say exactly
what we want, they give the hard direction in their Lemma 4.  Should
we just cite \cite{FLN+02} Lemma 4 for this?}}
For all $f:[m]^n\to\bits$ and $\eps >0$, we have that $f$ is $\eps$-far from monotone if and only if there exists $\eps m^n$ many pairwise disjoint ordered pairs of vertices $(x^i,y^i) \in [m]^n\times [m]^n$ such that $x^i  \prec y^i$ and $f(x^i) > f(y^i)$.  We will call each such pair a \emph{violation with respect to $f$}.
\end{theorem}

\begin{proof}[Proof of Proposition~\ref{distance-preserving}]
For every $f:\bn\to\bits$, we define $\Phi[f] : [m]^n\to\bits$  to be
the function
\begin{equation} \label{eq:Phim}
\Phi[f](x_1,\ldots,x_n) := f\big(\ind[x_1 > m/2],\ldots,\ind[x_n > m/2]\big),
\end{equation}
where we use $\ind[\cdot]$ to denote the $\{\pm 1\}$-valued indicator where $\ind[P] = 1$ if $P$ is true, and $-1$ otherwise. (Note that $m/2$ is an integer by our assumption that $m$ is even.)

It is straightforward to verify that $\Phi[f]$ is monotone if $f$ is monotone, and so it remains to show that $\Phi[f]$ is $\eps$-far from monotone if $f$ is $\eps$-far from monotone. Since $f$ is $\eps$-far from monotone, we have by Theorem~\ref{violating-pairs} that there exist $\eps 2^n$ many pairwise disjoint pairs $(x^i,y^i) \in \bn\times\bn$ that are violations with respect to $f$; we will exhibit $\eps m^n$ many pairwise disjoint pairs in $[m]^n$ that are violations with respect to $\Phi[f]$, which along with another application of Theorem~\ref{violating-pairs} completes the proof.  Let $S:\bits \to \big\{[m/2], \{(m/2)+1,\ldots,m\}\big\}$ be the set-valued function
\[ S(b) = \left\{
\begin{array}{cl}
[m/2] & \text{if $b = -1$} \\
\{(m/2)+1,\ldots,m\} & \text{if $b = 1$},
\end{array}
\right.
\]
and by a slight abuse of notation, we also define
\[ S(x) = S(x_{1}) \times \cdots \times S(x_n) \subseteq [m]^n \]
to be a function that maps points $x\in \bn$ to subsets of $[m]^n$.  Note that $|S(x)| = (m/2)^n$ for all $x\in \bn$, and  $S(x) \cap S(y) = \emptyset$ if $x\ne y$. Furthermore, $\Phi[f](x') = f(x)$ for all $x \in \bn$ and $x'\in S(x)$. In words, $S$ maps each $1$-input of $f$ to a set of $(m/2)^n$ many $1$-inputs of $\Phi[f]$, and likewise each $0$-input of $f$ to a set of $(m/2)^n$ many $0$-inputs of $\Phi[f]$.

For any pair $(x,y)\in\bn\times\bn$ that is a violation with respect to $f$, consider pairing the $(m/2)^n$ elements of $S(x)$ with the $(m/2)^n$ elements of $S(y)$ in the obvious way (i.e. each $a=(a_1,\dots,a_n) \in S(x)$ is
is paired with the unique element $b=(b_1,\dots,b_n) \in S(y)$
that has $(a_i \mod m/2)=(b_i \mod m/2)$ for all $i$).
\ignore{\rnote{This was ``in lexicographic order'', I got confused thinking about what that meant exactly.}}
Since $x\prec y$, it follows from the definition of $S$ that every $x'\in S(x)$ is paired with $y'\in S(y)$ where $x'\prec y'$. Furthermore, as noted above $\Phi[f](x') = f(x) = 1$ whereas $\Phi[f](y') = f(y) = 0$, and so every pair $(x',y')\in S(x)\times S(y)$ is a violation with respect to $\Phi[f]$.  Therefore each of the $\eps 2^n$ many pairs $(x,y)\in \bn\times\bn$ that are violations with respect to $f$ gives rise to $(m/2)^n$ many pairwise disjoint pairs $(x',y')\in S(x) \times S(y)$ that are violations with respect to $\Phi[f]$. Finally recalling that $S(x) \cap S(y) = \emptyset$ if $x \ne y$, we conclude that there are indeed $\eps 2^n\cdot (m/2)^n = \eps m^n$ many pairwise disjoint pairs that are violations with respect to $\Phi[f]$. This finishes the proof.
\end{proof}

\begin{remark}
The proof for odd $m$, deferred to Appendix~\ref{app:odd-m}, is via a similar but more involved version of Proposition~\ref{distance-preserving}.  In place of the simple indicator function
$\ind[x_i > m/2]$ (whose domain is simply $[m]$), we now use an ``almost-balanced''
monotone function $h: [m]^{k} \to \{-1,1\}$ where $k=\Theta(\log n)$ and $h$ has
some additional properties.  The fact that $k=\Theta(\log n)$ incurs an additional
logarithmic loss in the parameters but still results in a $\tilde{\Omega}(n^{1/5})$ lower
bound.
\end{remark}

\section{The algorithm}
\label{sec-alg}
\newcommand{\dens}{\mathrm{dens}}

\newcommand{\xx}{\bx}
\newcommand{\yy}{\by}
\newcommand{\zz}{\bz}

Throughout the proof of our upper bound we will assume that $1/n \le \eps \le 1/2.$ Note that this is without loss of generality, since if $\eps < 1/n$ then
the edge tester alone succeeds with probability $\Omega(\eps/n) = \Omega(\eps^2)$, and if $\eps > 1/2$ then every $f$ is $\eps$-close to
one of the two constant functions, both of which are monotone.

For our upper bound
it will be more convenient to view Boolean functions as mapping $\{0,1\}^n$ to $\{0,1\}$.  For $x,y
\in \{0,1\}^n$ we write $\|x\|_1$ to denote $\sum_{i=1}^n x_i$, the number of $1$s in $x$, and $\|x-y \|_1$ to denote $|\hspace{0.02cm}\{ i \in [n] \colon x_i \ne y_i\}\hspace{0.02cm}|$,
  the $\ell_1$-distance between $x$ and $y$. Given $1/n\le \eps \leq 1/2$, we
fix \[ d(n,\eps) := 2 \left\lceil \sqrt{2n\ln(100/\eps)}\right\rceil=
  O\big(\sqrt{n\ln(1/\epsilon)}\big),\]
and will denote $d(n,\eps)$ simply by $d$ when
  the distance parameter $\eps$ is clear from the context.
For each $i \in \{0,1,\ldots,n\}$ we let $L_i := \{x \in \zo^n \colon \|x\|_1 = i\}$ denote the $i$-th layer,
  and refer to
$$L_{\mathrm{mid}} :=  \big\{ x\in  L_i \colon i\in  [ (n-d)/2, (n+d)/2]\hspace{0.03cm}\big\}$$
as the middle layers of the hypercube. A standard Chernoff bound gives
$|\{0,1\}^n \setminus L_{\mathrm{mid}}\hspace{0.02cm}| \leq (\eps/50)\hspace{0.03cm}2^n.$
Finally, by a ``path'' we always mean a directed path of $n+1$ adjacent vertices from $0^n$ up to $1^n$.

\subsection{Two useful distributions over comparable pairs}
\label{sec-two-useful-dist}

Let $\calD = \calD_{n,\eps}$ denote the following distribution over comparable pairs $(\xx,\yy) \in L_{\mathrm{mid}}\times L_{\mathrm{mid}}$: \begin{enumerate}
\item First pick a path $\bp$ uniformly from the collection of all paths going from $0^n$ to $1^n$.\vspace{-0.1cm}
\item Pick $\xx$ and $\yy$ independently and uniformly from $\bp_{\mathrm{mid}} := \{ z \in\bp \colon z \in L_{\mathrm{mid}}\}$.
\end{enumerate}
This distribution is a slight variant of the one induced by the~\cite{CS13a} path tester, which takes a parameter $\sigma$ as input and disallows pairs $(x,y)$ for which $\| x - y\| _1$ is too small relative to $\sigma$. Our new tester will \emph{not} sample from $\calD$ (see Section~\ref{our-path-tester}), but we will use $\calD$ in our analysis. We remark here that $\xx = \yy$ with positive probability under $\calD$.

If $\xx,\yy$ were chosen independently and uniformly from $\{0,1\}^n$, then the probability that they both land in a fixed set $A$ of $\sigma 2^n$ points, for some $\sigma\in (0,1)$, would be $\sigma^2$.  The following lemma states that the probability is not much lower for
a pair drawn from $\calD$:

\begin{lemma}\label{BBL}
Let $A \subseteq L_{\mathrm{mid}}$ be a set of $\sigma2^n$ points.  Then $\Prx_{(\xx,\yy)\leftarrow \calD} [\xx,\yy \in A ] = \Omega\hspace{0.03cm} \big(\sigma^2\ln^{-1}(1/\eps) \big).$
\end{lemma}

\begin{proof}
Applying Jensen's inequality, we have
\begin{align*} \Prx_{(\xx,\yy)\leftarrow\calD} [\xx,\yy\in A ]\hspace{0.03cm}
&=\hspace{0.03cm}\Ex_{\bp} \Big[\Prx_{\xx,\yy\in \bp} [\xx,\yy \in A ] \Big] \\
&=\hspace{0.03cm}\Ex_{\bp} \left[\left(\frac{|\bp_{\text{mid}} \cap A|}{|\bp_{\text{mid}}|} \right)^2\right] =\hspace{0.03cm}\Omega\left(\frac1{n\ln(1/\eps)}\right) \cdot\Ex_{\bp} \big[|\bp_{\text{mid}}\cap A|\big]^2,
\end{align*}
and so it suffices to lower bound $\Ex_{\bp}[|\bp_{\text{mid}} \cap A|]$ by $\Omega(\sigma\sqrt{n})$. This is exactly Claim 2.2.1 of~\cite{CS13a}; we repeat the calculation here for the sake of completeness:
\begin{align}
\Ex_{\bp}\big[|\bp_{\text{mid}} \cap A|\big] \hspace{0.03cm}&=\hspace{0.03cm} \Ex_{\bp}\left[\sum_{i=\frac1{2}(n-d)}^{\frac1{2}(n+d)} \ind\big[(\bp_{\text{mid}} \cap L_i) \sse A\big]\right] \nonumber \\
&=\hspace{0.02cm} \sum_{i=\frac1{2}(n-d)}^{\frac1{2}(n+d)} \Ex_{\bp}\Big[\ind\big[(\bp_{\text{mid}} \cap L_i) \sse A\big]\Big] \nonumber \\
&=\hspace{0.02cm} \sum_{i=\frac1{2}(n-d)}^{\frac1{2}(n+d)}  \frac{|A \cap L_i|}{|L_i|} \label{eq:BBL1} \\
&\ge \hspace{0.03cm} \frac{\sqrt{n}}{2^n} \sum_{i=\frac1{2}(n-d)}^{\frac1{2}(n+d)} |A \cap L_i| \hspace{0.03cm}= \hspace{0.03cm} \frac{|A| \sqrt{n}}{2^n}  =\hspace{0.03cm} \sigma\sqrt{n}, \label{eq:BBL2}
\end{align}
where we use $\ind[\cdot]$ to denote the $\{0,1\}$-valued indicator where $\ind[P] = 1$ if $P$ is true, and $0$ otherwise.
Here (\ref{eq:BBL1}) uses the fact that a uniformly random path $\bp$ from $0^n$ to $1^n$ contains a uniformly random point in layer $L_i$, and (\ref{eq:BBL2}) holds
since $|L_i|\le 2^n/\sqrt{n}$ for all $i$.
\end{proof}

We will need a numerical lemma concerning the ratio of binomial coefficients.

\begin{lemma}
\label{binomial-ratio}
Let $\eps \ge 1/n$, and $a,b\in [(n-d)/2,(n+d)/2]$ be integers where $a > b$. Then
\[  {a \choose a-b}\Big/{n-b\choose a-b} = O(1/\eps^4)\quad\ \text{and}\ \quad
{n\choose n/2}\Big/{n \choose a}=O(1/\epsilon^4). \label{eq:ratio}
\] \end{lemma}
\begin{proof}
We prove the first equation and the second equation is similar.  By a routine calculation we verify that the first ratio is maximized when $a = (n+d)/2$ and $b = n/2$, and so \[  \frac{{a \choose a-b}}{{n-b\choose a-b}} \le  \frac{\frac1{2}(n+d)}{\frac{n}{2}}\cdot \frac{\frac1{2}(n+d)-1}{\frac{n}{2}-1} \cdots \frac{\frac{n}{2}+1}{\frac1{2}(n-d)+1}\\
\le \exp\left(\sum_{i=0}^{(d/2)-1} \frac{d/2}{(n/2)-i}\right),
\]
where we used $(1+t) \le e^t$ for $t \in \R$.
The lemma follows from the definition of $d$ and $\epsilon\ge 1/n$.
\end{proof}

For our analysis, the following distribution $\calD'=\calD'_{n,\epsilon}$ over
  comparable pairs $(\bx,\by)\in L_{\mathrm{mid}}\times L_{\mathrm{mid}}$
  in the middle layers comes in handy:
\begin{enumerate}
\item First pick a point $\bx$ uniformly at random from $L_{\mathrm{mid}}$.\vspace{-0.08cm}
\item Then pick a path $\bp$ uniformly from the collection of all paths going
  through $0^n$, $\bx$, and $1^n$.\vspace{-0.08cm}
\item Pick $\yy$ uniformly from $\bp_{\mathrm{mid}} := \{ z \in\bp \colon z \in L_{\mathrm{mid}}\}$.
\end{enumerate}
We note that $\calD'$ is not the same as $\calD$, since picking a uniformly random $\bx$ from the middle layers of a uniformly random path $\bp$ does not induce a uniform distribution over $L_{\text{mid}}$; however, Lemma~\ref{binomial-ratio} allows us to switch between these essentially-equivalent distributions at the cost of a $O(1/\eps^4)$ factor. (On the other hand the conditional distributions $\calD_{\bx = x}$ and $\calD'_{\bx = x}$ on $\by$ are the same for all possible outcomes $x \in L_{\mathrm{mid}}$ of $\bx$.)

We get the following corollary from Lemmas \ref{BBL} and \ref{binomial-ratio}:
\begin{corollary}\label{hhh}
Let $A \subseteq L_{\mathrm{mid}}$ be a set of $\sigma2^n$ points. Then
$$\Prx_{(\bx,\by)\leftarrow \calD'} [\bx,\by \in A ] = \Omega\hspace{0.03cm} \big(\sigma^2\epsilon^4\ln^{-1}(1/\eps) \big).$$
\end{corollary}
\begin{proof}
It is clear from the definition of $\calD,\calD'$ that the conditional distribution of
  $\yy$ induced from $\calD$ by conditioning on a particular outcome of $\bx$
  is the same as that induced from $\calD'$ under the same conditioning.
It follows from the second part of Lemma \ref{binomial-ratio} that for any $x\in L_{\mathrm{mid}}$ we have
$$
\Prx_{(\bx,\by)\leftarrow \calD'}[\bx=x]=\Omega(\epsilon^4)\cdot \Prx_{(\xx,\yy)\leftarrow \calD}[\xx=x].
$$
As a result, we have for every comparable pair $(x,y)$ in the middle layers
$$
\Prx_{(\bx,\by)\leftarrow \calD'}[(\bx,\by)=(x,y)]=\Omega(\epsilon^4)\cdot \Prx_{(\bx,\by)\leftarrow \calD}[(\bx,\by)=(x,y)].
$$
The claim then follows from Lemma \ref{BBL}.
\end{proof}

\subsection{Density and score}
\label{sec-alg-score}
We need the following definition to give a more detailed
  analysis on the consequence of Corollary \ref{hhh}, which is
  key to the analysis of our monotonicity tester described in Section \ref{tester-sec}.

\begin{definition}[density and score]
Let $A\subseteq \{0,1\}^n$. For all $x \in \zo^n$ and $k\in \{0,1,\ldots,n\}$, we define the
  following quantities:
\[ \dens^\down_k(x,A) :=
\Prx_{\substack{\yy\preceq x \\ \|\yy - x\|_1 = k}}[\yy \in A] \text{\ if $k \le \|x\|_1$, and $\dens^\down_k(x,A) := 0$ otherwise},
\]
and similarly
\[ \dens^\up_k(x,A) := \Prx_{\substack{\yy \succeq x\\ \|\yy - x\|_1 = k}}[\yy \in A] \text{\ if $k \le n-\|x\|_1$, and $\dens^\up_k(x,A) := 0$ otherwise}.\]
We also define
\[ \score^\down(x,A) :=  \sum_{k = 0}^n \dens^\down_k(x,A)\quad\ \text{and}\ \quad \score^\up(x,A) :=  \sum_{k = 1}^n \dens^\up_{k}(x,A), \]
and refer to $\score^\down(x,A)$ as the \emph{downward $A$-score of $x$} and $\score^\up(x,A)$ as its \emph{upward $A$-score}.
\end{definition}

We point out the asymmetry between the definitions of $\score^\down(x,A)$ and $\score^\up(x,A)$: the first is summed over $k$ starting at $0$, whereas the second is summed over $k$ starting at $1$. (Note that $\dens^\down_0(x,A) = \dens^\up_0(x,A) = \ind[x \in A]$.) We will need the fact that both the upward and downward $A$-scores of any $x\in \zo^n$ are at most $d=d(n,\eps)$ when $A\sse L_{\text{mid}}$.

The following lemma relates the distribution $\calD'$ (more precisely, the distribution over $\yy$ that is induced by conditioning on a particular outcome of $\xx$) to the notion of score:

\begin{lemma}
\label{score}
Let $A \subseteq L_{\mathrm{mid}}$ be a set of $\sigma2^n$ points and fix $x^* \in L_{\mathrm{mid}}$. Then
\[ \Prx_{(\xx,\yy)\leftarrow\calD'}\big[\yy \in A \mid \xx = x^*\big] = \frac1{\Theta(\sqrt{n\ln(1/\eps)})} \left({\score^\down(x^*,A) + \score^\up(x^*,A)}\right). \]
\end{lemma}

\begin{proof}
This holds since
\begin{align*}
\Prx_{(\xx,\yy)\leftarrow\calD'}\big[\yy \in A \mid \xx = x^*\big] &= \Ex_{\bp\ni x^*} \Bigg[ \frac{|\bp_{\text{mid}} \cap A|}{|\bp_{\text{mid}}|} \Bigg] \nonumber \\
&=
\frac1{\Theta(d)}\cdot \Ex_{\bp\ni x^*} \big[| \bp_{\text{mid}} \cap A|\big]
\\ &=
\frac1{\Theta(d)} \left( \sum_{k\ge 0} \left(\Ex_{\substack{\yy \preceq x^*\\ \|\yy-x^*\|_1=k}}\big[\ind[\yy\in A]\big]\right)
  +\sum_{k\ge 1} \left(\Ex_{\substack{\yy\succ x^*\\ \|\yy-x^*\|_1=k}}\big[\ind[\yy\in A]\big]\right)\right)
\\&=
\frac1{\Theta(d)} \left(\sum_{k\ge 0}
\left(\Prx_{\substack{\by \preceq x^* \\ \| \by - x^*\|_1 = k}}[\by \in A]\right)+\sum_{k\ge 1}\left(
\Prx_{\substack{\by\succ x^* \\ \| \by - x^*\|_1 = k}}[\by \in A]
\right)\right)\\&=
\frac1{\Theta(d)}\left(\score^\down(x^*,A)+\score^\up(x^*,A)\right).\qedhere
\end{align*} \renewcommand{\qedsymbol}{}
\end{proof}

We use the previous two lemmas to lower bound the expected downward $A$-score of an $\bx$
  drawn uniformly at random from $A$:

\begin{lemma}
\label{downward-score}
Let $\eps \ge 1/n$ and $A \sse L_{\mathrm{mid}}$ be a set of $\sigma2^n$ points. Then \[ \Ex_{\bx\in A}\big[\score^\down(\bx,A)\big] = \Omega\left(\dfrac{\eps^8\sigma\sqrt{n}}{\sqrt{\ln(1/\eps)}}\right).\]
\end{lemma}

\begin{proof}
We begin with the claim that
\begin{equation} \Ex_{\bx\in A}\big[\score^\down(\bx,A)\big]\ge \Omega(\eps^4) \Ex_{\bx\in A}\big[\score^\up(\bx,A)\big]+1, \label{up-vs-down}
\end{equation}
where the $+1$ is due to $\dens^{\down}_0(\bx,A)=1$.
To see (\ref{up-vs-down}), we rewrite the LHS of the inequality as follows:
\begin{align}
\Ex_{\bx\in A}\big[\score^\down(\bx,A)\big]-1 \hspace{0.03cm}
&=\hspace{0.03cm} \frac1{\sigma 2^n} \sum_{x\in A} \sum_{k\ge 1}
\sum_{\substack{y \prec x \\ \|y-x\|_1 = k}} \frac{\ind[y\in A]}{{\|x\|_1\choose k}}\nonumber \\
&=\hspace{0.03cm} \frac1{\sigma 2^n} \sum_{x\in A} \sum_{\substack{y\in A\\ y \prec x}} \frac1{{\|x\|_1\choose \| x -y\|_1}} \nonumber \\
&=\hspace{0.03cm} \frac1{\sigma 2^n} \sum_{y\in A} \sum_{\substack{x\in A\\ x\succ y}} \frac{{n-\|y\|_1\choose \| x - y\|_1}}{{\|x \|_1 \choose \| x - y \|_1}} \cdot \frac1{{n-\|y\|_1\choose \| x - y\|_1}} \nonumber \\
&\ge\hspace{0.03cm}  \min_{\substack{x \succ y\\ x,y\in L_{\mathrm{mid}}}}\Bigg\{ \frac{{n-\|y\|_1\choose \| x - y\|_1}}{{\|x \|_1 \choose \| x - y \|_1}}\Bigg\} \Ex_{\by\in A} \big[\score^\up(\by,A)\big]
\hspace{0.03cm} = \hspace{0.03cm} \Omega(\eps^4) \Ex_{\by\in A}\big[\score^\up(\by,A)\big], \nonumber
\end{align}
where the final equality holds by the first part of Lemma~\ref{binomial-ratio}. This proves (\ref{up-vs-down}), which together with Lemma~\ref{score} gives
\begin{align} \Prx_{(\bx,\by)\leftarrow\calD'}\big[\by\in A \mid \bx\in A\big]\hspace{0.02cm} &= \hspace{0.02cm}\frac1{\Theta(\sqrt{n\ln(1/\eps)})} \Ex_{\bx\in A} \big[\score^\up(\bx,A) + \score^\down(\bx,A)\big]\nonumber \\[0.5ex]
&=\hspace{0.02cm}  \frac{O(\eps^{-4})}{\Theta(\sqrt{n\ln(1/\eps)})} \left(\Ex_{\bx\in A}\big[\score^\down(\bx,A)\big]\right). \label{up-vs-down3}
\end{align}
On the other hand, by Corollary~\ref{hhh}\ignore{\rnote{This was ``Lemma~\ref{BBL}'' but I think it's the corollary we're using instead.}} we have
\begin{align}
\Prx_{(\bx,\by)\leftarrow\calD'}\big[\by\in A \mid \bx \in A\big] \hspace{0.02cm}
&= \hspace{0.02cm}\frac{\Prx_{(\bx,\by)\leftarrow \calD'} [\bx,\by\in A]}{\Prx_{(\bx,\by)\leftarrow \calD'} [\bx\in A]} \nonumber \\[0.3ex]
&=\hspace{0.02cm}  \frac{\Prx_{(\bx,\by)\leftarrow \calD'} [\bx,\by\in A]}{\sigma}\ =\ \Omega\left(\frac{\epsilon^4\sigma}{\ln(1/\eps)}\right).\label{up-vs-down4}
\end{align}
Combining (\ref{up-vs-down3}) with (\ref{up-vs-down4}) and rearranging completes the proof.
\end{proof}


Lemma \ref{downward-score} lower bounds the average downward $A$-score of points $x \in A$; its conclusion may be equivalently rewritten as the following sum:
\begin{equation} \label{eq:lb}
\sum_{x \in A} \score^\down(x,A) = \Omega\left(\frac{\eps^8 \sigma^2 \sqrt{n} \hspace{0.04cm}2^n}{\sqrt{\ln(1/\epsilon)}}\right).
\end{equation}

We may express the downward $A$-score $\score^{\down}(x,A)$ of a point $x$ as a sum over $m+1$  ``buckets'' of exponentially increasing size:
\begin{equation}\label{eq:haha}
\score^{\down}(x,A) = \left(\sum_{k \in B_0} \dens^\down_{k}(x,A) \right)+ \left(\sum_{k \in B_1} \dens^\down_{k}(x,A) \right)+ \cdots + \left( \sum_{k \in B_m} \dens^\down_{k}(x,A)\right),
\end{equation}
where $B_0=\{0\}$ and $B_i=\{2^{i-1},\dots,2^{i}-1\}$ for each $i\in [m]$ and
  $m = \lceil \log (n+1)\rceil$.
It will be useful for us to focus on a particular bucket $\ell\in \{0,1,\ldots,m\}$
  such that the overall sum of $\score^\down(x,A)$ in (\ref{eq:lb}) has a ``large'' contribution
  from the $\ell$-th bucket.
A straightforward argument, exploiting the fact
that there are only logarithmically many buckets, lets us achieve this without losing too much in the sum: 

\begin{corollary}
\label{bucket}
Let $\eps \ge 1/n$ and $A\sse L_{\mathrm{mid}}$ be a set of {$\sigma 2^n$ points}.
There exists
  $\ell \le m$ 
  such that
\begin{equation}\label{eq:ttt} \sum_{x\in A}\sum_{k \in B_\ell} \dens^\down_k(x,A)  = \Omega\left(\frac{ \eps^8\sigma^2\sqrt{n}\hspace{0.04cm}2^n}{(\log n)\sqrt{\ln(1/\eps)}}\right). \end{equation}
\end{corollary}
\begin{proof}
This follows from (\ref{eq:lb}), (\ref{eq:haha}), and the fact that there are only $m+1$ many buckets.
\end{proof}

Corollary~\ref{bucket} gives us a lower bound on the sum of downward $A$-scores of points $x \in A$ coming from a certain bucket $B_\ell$. Our next corollary uses this to give a lower bound on the sum of downward $A$-scores of points $y \in A_u$ coming from (essentially) the same bucket $B_\ell$, where $A_u$ is an ``upper vertex boundary'' of $A$ in the following sense: there exists an $|A|$-sized matching $M$ of edges $(x,y)$ where $x\prec y$, $x \in A$ and $y \in A_u$.

\begin{corollary}
\label{bucket2}
Let $\eps \ge 1/n$ and $M$ be a matching of $\sigma 2^n$ edges in the middle layers. Let
\begin{align*}
A&:= \big\{ x \in \zo^n\colon \text{$x \prec y$ and $(x,y)\in M$}\big\}\quad \text{and}\\[0.6ex]
A_u &:= \big\{ y \in \zo^n \colon \text{$y \succ x$ and $(x,y)\in M$}\big\}
\end{align*}
be the lower and upper endpoints of edges in $M$,
  respectively.  For each bucket $B_i$, $i\in \{0,1,\ldots,m\}$, we let $B_i':=\{j+1:j\in B_i\}$. Then there exists
an integer $\ell\le m$ such that
\begin{equation}
\sum_{y\in A_u}\sum_{k\in B_\ell'}\dens^\down_k(y,A)  = \Omega\left(\frac{2^{\ell+n}\eps^8\sigma^2}{(\log n) \sqrt{n\ln(1/\eps)}}\right).\label{large-score}
\end{equation}
\end{corollary}

\begin{proof}
By Corollary~\ref{bucket}, there exists an $\ell\le m$ such that $A$ satisfies (\ref{eq:ttt}).

Next for every edge $(x,y) \in M$ 
  we have that
\[ \dens^\down_{k+1}(y, A) = \Prx_{\substack{\bz \prec y\\ \| \bz-y\|_1 = k+1}}[\bz \in A]  \ge
{\frac{{\|x\|_1\choose k}}{{\|y\|_1\choose k+1}}} \Prx_{\substack{\bz\prec x\\ \| \bz-x\|_1 = k}}[\bz \in A]  = \frac{(k+1)\cdot \dens^\down_{k}(x, A)}{\|x\|_1+1}. \]
Therefore, by (\ref{eq:ttt}) we have
\begin{align*}
\sum_{y\in {A_u}}\sum_{k\in B_\ell'} \dens^\down_k(y,A)
\hspace{0.02cm}&=\hspace{0.02cm}\sum_{y\in {A_u}}\sum_{k\in B_\ell} \dens^\down_{k+1}(y,A)\\
 &\ge\hspace{0.02cm} \sum_{x\in {A_u}}\sum_{k\in B_\ell} \frac{(k+1)\cdot \dens^\down_{k}(x,A)}{{\|x\|_1+1}} \\[0.36ex]
 &=\hspace{0.03cm} \Omega\left(\frac{ \eps^8\sigma^2\sqrt{n}\hspace{0.04cm}2^n}{(\log n)\sqrt{\ln(1/\eps)}}\cdot \frac{2^\ell}{n}\right).
\end{align*}
This completes the proof.
\end{proof}

\subsection{The weighted path tester and its analysis}\label{tester-sec}
\label{our-path-tester}
Given a Boolean function $f$, recall that a pair $(x, y)$ of vertices is a \emph{violated pair with respect to $f$ if $x \prec y$ and $f(x) > f(y)$.} Our algorithm $\texttt{weighted-path-tester}$ for monotonicity
  testing proceeds as follows:

\begin{framed}
$\texttt{weighted-path-tester}$:\vspace{-0.1cm}
\begin{enumerate}
\item Pick a point $\by$ uniformly from $L_{\mathrm{mid}}$.\vspace{-0.1cm}
\item Pick $\boldsymbol{\ell} \in \{0,1,\ldots,m=\lceil \log (n+1) \rceil \}$ uniformly.\vspace{-0.1cm}
\item Pick $\bk\in B_{\boldsymbol{\ell}}'$ uniformly.\vspace{-0.1cm}
\item Pick a path $\bp$ uniformly from the collection of all paths going through $0^n, \by$ and $1^n$, and set $\bx$ to be the
(unique) point on $\bp$ that has $\bx \prec \by$ and $\|\bx - \by\|_1=k.$\vspace{-0.1cm}
\item Reject iff $(\bx,\by)$ is a violated pair.
\end{enumerate}
\vskip -.13in
\end{framed}

We note that
an equivalent formulation of step (4) is that $\bx$ is drawn uniformly from $\big\{z \in \{0,1\}^n \colon \text{$z \prec \by$ and $\| \by-z\|_1 =\bk$} \big\}.$
Below we show that if there is a $(\sigma 2^n)$-sized matching $M$ of violated edges of $f$
  in the middle layers of the hypercube, then the tester above succeeds in finding a violated pair with probability
  roughly $\Omega(\sigma^2/\sqrt{n})$.

\begin{proposition}\label{theo:mainalg}
Let $f:\zo^n\to\zo$ and $\eps \ge 1/n$. Suppose there is a $(\sigma 2^n)$-sized matching $M$ of violated edges of $f$ all lying in the middle
layers of the hypercube. Then {\tt weighted-path-tester} succeeds (i.e.~samples $\bx$ and $\by$ that form a violated pair with respect to $f$) with probability
\begin{equation} \label{tester-prob}
\Omega\left(\frac{\epsilon^8\sigma^2}{(\log^2 n) \sqrt{n\ln(1/\epsilon)}}\right).
\end{equation}
\end{proposition}

\begin{proof}
Let $A$ be the $1$-endpoints of edges in $M$, and $A_u$ be the $0$-endpoints, and note that every  pair $(x,y) \in A\times A_u$ satisfying $x\prec y$ is a violated pair with respect to $f$. Let $\calD^w$ denote the distribution over comparable pairs $(\bx,\by)\in L_{\text{mid}}\times L_{\text{mid}}$ induced by {\tt weighted-path-tester}.  Applying Corollary~\ref{bucket2}, we know there that exists an $\ell^*\in \{0,1,\ldots,m\}$ such that \[
\sum_{y\in {A_u}}\sum_{k\in B_{\ell^*}'} \dens^\down_k(y,A)
 = \Omega\left(\frac{ 2^{\ell^*+n}\eps^8\sigma^2}{(\log n)\sqrt{n\ln(1/\eps)}}\right).
\]

Note that conditioning on the event of $\yy=y$ and $\bk=k$, 
  the probability of $\xx\in A$ is $\dens^{\down}_k(y,A)$.
Since $\by,\boldsymbol{\ell},\bk$ are all sampled uniformly,
  $\texttt{weighted-path-tester}$ succeeds with probability at least
\begin{align*}
{\Prx_{(\bx,\by)\leftarrow \calD^w}\big[\by \in A_u, \bx \in A\big]}\hspace{0.02cm}
&= \hspace{0.03cm}{\Prx_{(\bx,\by)\leftarrow \calD^w}[\by\in A_u] \cdot  \Prx_{(\bx,\by)\leftarrow \calD^w}\big[\bx \in A\mid \by\in A_u\big]} \\
&= \hspace{0.03cm}{\frac{|A_u|}{|L_{\text{mid}}|} \cdot \frac1{|A_u|} \sum_{y\in A_u} \frac1{m+1}\sum_{\ell = 0}^m \frac1{|B_\ell'|}\sum_{k\in B'_\ell} \dens^\down_k(y,A)} \\[-0.3ex]
&\ge \hspace{0.03cm}\frac{1}{(m+1)\hspace{0.03cm}|L_{\mathrm{mid}}|\hspace{0.03cm}|B_{\ell^*}'|}\cdot\sum_{y\in {A_u}}\sum_{k\in B_{\ell^*}'}\dens^\down_k(y,A) \\
&=\hspace{0.03cm}\Omega\left(\frac{ 2^{\ell^*+n}\eps^8\sigma^2}{(\log n)\sqrt{n\ln(1/\eps)}}
\cdot \frac{1}{(\log n)\hspace{0.03cm}2^{\ell^*+n} } \right)\hspace{0.03cm}{=\hspace{0.03cm}\Omega\left(\frac{ \eps^8\sigma^2}{(\log^2 n)\sqrt{n\ln(1/\eps)}}\right).}
\end{align*}
This finishes the proof.
\end{proof}

\subsection{Proof of Theorem~\ref{main-theorem-alg}}
\label{sec-alg-proof}
Finally we combine Proposition \ref{theo:mainalg} with the dichotomy theorem of \cite{CS13a}
  to prove Theorem \ref{main-theorem-alg}.
To state the latter, we let $v 2^n$ denote the total number of
  violated edges in $f$. We also let $\sigma 2^n$ denote the size of the largest
  matching of violated edges in the middle layers.
Then we have

\begin{theorem}[Theorem 2.4 of \cite{CS13a}]\label{theo-dichotomy}
For any $f$ that is $\epsilon$-far from monotone, $v\cdot \sigma=\Omega(\epsilon^2)$.
\end{theorem}

We now prove Theorem \ref{main-theorem-alg}.

\begin{proof}[Proof of Theorem \ref{main-theorem-alg}]
As mentioned at the beginning of Section~\ref{sec-alg}, we may assume without
  loss of generality that $\eps \ge 1/n$ since otherwise the edge tester alone succeeds with probability $\Omega(\eps/n) = \Omega(\eps^2)$.
  When $\eps \ge 1/n$, our tester flips a coin, runs the edge tester with probability
  $1/2$, and~runs $\texttt{weighted-path-tester}$ with probability $1/2$.
Given $v$ and $\sigma$ as defined above,
  the success probability of the edge tester is $\Omega(v/n)$;
  the success probability of $\texttt{weighted-path-tester}$ is given in (\ref{tester-prob}).
It follows from Theorem \ref{theo-dichotomy} that the average of these two
  is at least
\[
\Omega\left(\frac{\epsilon^4}
{n^{5/6}(\log^{2/3}n) \hspace{0.04cm}(\ln (1/\epsilon))^{1/6}}\right).\]
This finishes the proof of Theorem \ref{main-theorem-alg}.
\end{proof}

\section*{Acknowledgements}

We thank Eric Blais for a helpful discussion that led to an improvement of our lower bound for general hypergrid domains.

\appendix 

\bibliography{odonnell-bib}{}
\bibliographystyle{alpha}

\section{Reduction from hypergrid domains $[m]^n$ when $m$ is odd} 
\label{app:odd-m} 
\begin{lemma}
\label{odd-m-lemma} 
Let $m\in \N$ be odd. There exists a monotone function $h:[m]^k\to\bits$ such that $|\{ x \in [m]^k \colon h(x)=1\}| = |\{ x \in [m]^k \colon h(x)=-1\}| + 1$, and a one-to-one mapping $\Psi : h^{-1}(-1) \to h^{-1}(1)$ such that $\Psi(x) \succ x$ for all $x \in h^{-1}(-1)$. 
\end{lemma} 

\begin{proof}
The function $h$ is defined as follows: 
\[ 
h(x_1,\ldots,x_k) = \left\{ 
\begin{array}{cl} 
1 & \text{if $x = \lceil m/2\rceil ^k$} \\
\sign(x_i-\lceil m/2\rceil) & \text{otherwise, where $i := \min\{ i \in [k]\colon x_i \ne \lceil m/2\rceil\}$.} 
\end{array} 
\right.
\] 
The monotonicity of $h$ is straightforward to verify, as is the fact that $|\{ x \in [m]^k \colon h(x)=1\}| = |\{ x \in [m]^k \colon h(x)=-1\}|  + 1$. The proof is complete by noticing that the mapping
\[ \Psi(x_1,\ldots,x_k) = (x_1,\ldots,x_{i-1},-x_i,x_{i+1},\ldots,x_k), \quad \text{where $i := \min\{ i \in [k] \colon x_i \ne \lceil m/2\rceil \}$} \]
is a bijection between $h^{-1}(-1)$ and $h^{-1}(1) \setminus \{ \lceil m/2\rceil^k\}$. 
\end{proof} 

With Lemma~\ref{odd-m-lemma} in hand we are ready to prove the following analogue of Proposition~\ref{distance-preserving} for hypergrid domains $[m]^n$ when $m$ is odd. Given the monotone function $h$ defined in Lemma~\ref{odd-m-lemma}, let $h':[m]^k \to \{-1,1,\bot\}$ be the partial function where $h'(x) = \bot$ if $x = \lceil m/2\rceil^k$, and $h'(x) = h(x)$ otherwise (and so  $|\{ x \in [m]^k \colon h(x)=1\}| = |\{ x \in [m]^k \colon h(x)=-1\}| = (m^k-1)/2$). 

\begin{proposition}
\label{prop:odd-m} 
For all odd $m\in \N$ the mapping
\[ \Phi : \big\{ \text{all functions $f:\{-1,1\}^n\to\bits$} \big\} \to \big\{  \text{all functions $f : [m]^{n \lceil \log n\rceil} \to\bits$}\big\} \]
defined by \emph{(\ref{eq:Phi-odd-m})} below satisfies the following two properties:
\begin{enumerate}
\item If $f\isafunc$ is monotone then $\Phi[f]$ is monotone as well.
\item If $f\isafunc$ is $\eps$-far from monotone then $\Phi[f]$ is $\Omega(\eps)$-far from monotone.
\end{enumerate}
\end{proposition} 

\begin{proof} 
Fix $k:= \lceil \log n\rceil$. For every $f:\bn\to\bits$, we define $\Phi[f] : [m]^{kn}\to\bits$  to be
the following function: for all $x^1,\ldots,x^n \in [m]^k$ 
\begin{equation} \label{eq:Phi-odd-m}
\Phi[f](x^1,\ldots,x^n) := f\big(h(x^1),\ldots,h(x^n)\big).
\end{equation}
Since $h$ is monotone it follows that $\Phi[f]$ is monotone if $f$ is monotone, and so it remains to show that $\Phi[f]$ is $\Omega(\eps)$-far from monotone if $f$ is $\eps$-far from monotone. Since $f$ is $\eps$-far from monotone, we have by Theorem~\ref{violating-pairs} that there exist $\eps 2^n$ many pairwise disjoint pairs $(x^i,y^i) \in \bn\times\bn$ that are violations with respect to $f$; we will exhibit $\Omega(\eps\, m^{kn})$ many pairwise disjoint pairs in $[m]^{kn}$ that are violations with respect to $\Phi[f]$, which along with another application of Theorem~\ref{violating-pairs} completes the proof.  

Using the same notation as in the proof of Proposition~\ref{distance-preserving}, we define the set-valued function $S$ mapping $x\in \bn$ to subsets of $[m]^{kn}$ as follows: 
\[ S(x) = h'^{-1}(x_1) \times \cdots \times h'^{-1}(x_n) \sse  [m]^{kn}.\]
Note that 
\[ |S(x)| = \left(\frac{m^k-1}{2}\right)^n  = \frac{m^{kn}}{2^n}\left(1-\frac1{m^k}\right)^n  \ge \frac{m^{kn}}{2^n}\left(1-\frac1{n^{\log m}}\right)^n  = \Omega\left(\frac{m^{kn}}{2^n}\right)\] 
for all $x \in \bn$ (where we have used our choice of $k=\lceil \log n\rceil$ for the inequality), and $S(x) \cap S(y) = \emptyset$ if $x \ne y$.  Furthermore, $\Phi[f](x') = f(x)$ for all $x \in \bn$ and $x'\in S(x)$. In words, $S$ maps each $1$-input of $f$ to a set of $((m^k-1)/2)^n$ many $1$-inputs of $\Phi[f]$, and likewise each $0$-input of $f$ to a set of $((m^k-1)/2)^n$ many $0$-inputs of $\Phi[f]$.

For any pair $(x,y)\in\bn\times\bn$, $x\prec y$, that is a violation with respect to $f$, consider pairing the $((m^k-1)/2)^n$ elements of $S(x)$ with the $((m^k-1)/2)^n$ elements of $S(y)$ via $\Psi$ from Lemma~\ref{odd-m-lemma} as follows: each $a \in S(x)$, which we will view as $a = (a_1,\ldots,a_n) \in ([m]^k)^n$, 
is paired with the unique element $b=(b_1,\dots,b_n) \in S(y)$ where $b_i = a_i$ if $x_i = y_i$, and $b_i = \Psi(a_i)$ if $x_i < y_i$.  Since $x\prec y$, it follows from the definitions of $S$ and $\Psi$ that every $x'\in S(x)$ is paired with $y'\in S(y)$ where $x'\prec y'$. Furthermore, as noted above $\Phi[f](x') = f(x) = 1$ whereas $\Phi[f](y') = f(y) = 0$, and so every pair $(x',y')\in S(x)\times S(y)$ is a violation with respect to $\Phi[f]$.  Therefore each of the $\eps 2^n$ many pairs $(x,y)\in \bn\times\bn$ that are violations with respect to $f$ gives rise to $\Omega(m^{kn}/2^n)$ many pairwise disjoint pairs $(x',y')\in S(x) \times S(y)$ that are violations with respect to $\Phi[f]$. Finally recalling that $S(x) \cap S(y) = \emptyset$ if $x \ne y$, we conclude that there are indeed $\eps 2^n\cdot \Omega(m^{kn}/2^n) = \Omega(\eps\, m^{kn})$ many pairwise disjoint pairs that are violations with respect to $\Phi[f]$. This finishes the proof.
\end{proof}

Proposition~\ref{prop:odd-m} along with Theorem~\ref{main-theorem-lb} implies the existence of a universal constant $\eps_0 > 0$ such that any non-adaptive $\eps_0$-tester for the monotonicity of $f: [m]^{N} \to \bits$, where $N := n\lceil \log n\rceil$ and $m$ is odd, must make $\tilde{\Omega}(n^{1/5}) = \tilde{\Omega}(N^{1/5})$ many queries. This along with Proposition~\ref{distance-preserving} (establishing the same lower bound for hypergrid domains $[m]^n$ where $m$ is even) completes the proof of Theorem~\ref{main-hypergrid}.

\end{document}